\documentclass{tlp-hacked}
\pdfoutput=1

\usepackage{amsmath}
\usepackage{amssymb} 
\usepackage{stmaryrd}
\usepackage{graphicx}
\usepackage{tikz}
\usetikzlibrary{arrows}
\usetikzlibrary{topaths}
\usetikzlibrary{calc}
\usepackage{ifthen} 
\usepackage{aliascnt} 
\usepackage{hyperref}
\hypersetup{
  breaklinks=true
}
\usepackage[sort]{cleveref} 
\usepackage{enumerate}

 \newcommand{\removed}[1]{{\color{gray}\xout{#1}}}

\newcommand\Prog{{\mathcal P}}

\renewcommand{\removed}[1]{}

 \newcommand\sstate[1]{\left< {#1} \right>}

\newcommand\wrt{w.r.t\@.}

\newcommand\Newtheorem[2]{
  \newaliascnt{#1}{theorem} 
  \newtheorem{#1}[#1]{#2}  
  \aliascntresetthe{#1} 
}

\newtheorem{theorem}{Theorem}
\Newtheorem{lemma}{Lemma}
\Newtheorem{proposition}{Proposition}
\Newtheorem{corollary}{Corollary}
\Newtheorem{definition}{Definition}
\Newtheorem{example}{Example} 
\Newtheorem{remark}{Remark}
\crefname{section}{Sect.}{Sect.}

\renewcommand\cref[1]{\Cref{#1}}

\newcommand\ra{\rightarrow}

\newcommand\ie{i.e\@.}

\newcommand\R{\mathcal R}

\newcommand\Y{{\cal Y}}

\renewcommand\AA{{\mathbb A}}
\newcommand\BB{{\mathbb B}}
\newcommand\CC{{\mathbb C}}
\newcommand\DD{{\mathbb D}}
\newcommand\EE{{\mathbb E}}
\newcommand\FF{{\mathbb F}}
\newcommand\GG{{\mathbb G}}
\newcommand\HH{{\mathbb H}}
\newcommand\KK{{\mathbb K}}

\newcommand\fv{\textup{fv}}

\newcommand\CT{{\mathcal{C}}}
\newcommand\CP{\CT\Prog}

\newcommand\lv{\textup{lv}}

\newcommand\at{{\;@\;}}

\renewcommand\Prog{{\mathcal P}}
\newcommand\Q{{\mathcal Q}}

\newcommand\simplif{~{\Longleftrightarrow}~}
\newcommand\propag{~{\Longrightarrow}~}

\newcommand\default[2]{\ifthenelse{\equal{#1}{}}{#2}{#1}}

\newcommand\estate[3][]{{{\ifthenelse{\equal{#1}{}}{}{\exists {#1}}}\left< {#2};{#3} \right>}}

\newcommand\state[3]{\langle \default{#1}{\emptyset};\default{#2}{\true}
;\default{#3}{\emptyset}
 \rangle}

\newcommand\true{\top}

\renewcommand\vec[1]{{\bar{#1}}}

\usepackage{amsfonts} 

\makeatletter
\def\twoheadrightarrowfill@{%
  \arrowfill@\relbar\relbar\twoheadrightarrow
}

\newcommand{\xtwoheadrightarrow}[2][]{%
   \ext@arrow  0359\twoheadrightarrowfill@{#1}{#2}%
}

\def\twoheadleftarrowfill@{%
  \arrowfill@\twoheadleftarrow\relbar\relbar
}

\newcommand{\xtwoheadleftarrow}[2][]{%
   \ext@arrow  0359\twoheadleftarrowfill@{#1}{#2}%
}

\makeatother

\newcommand\rap[1][\Prog]{\xrightarrow{_{#1}}}
\newcommand\rapsymb[2]{\xrightarrow{_{#1}}{\!\!}^{#2}\,}
\newcommand\rape[1][\Prog]{\rapsymb{#1}{\equiv}}
\newcommand\rapt[1][\Prog]{\rapsymb{#1}{\text{+}}}
\newcommand\raps[1][\Prog]{\xtwoheadrightarrow{_{#1}}}

\newcommand\lap[1][\Prog]{\xleftarrow{_{#1}}}
\newcommand\lapsymb[2]{\xleftarrow{_{#1}}{\!\!}^{#2}\,}
\newcommand\lape[1][\Prog]{\lapsymb{#1}{\equiv}}

 \newcommand\laps[1][\Prog]{\xtwoheadleftarrow{_{#1}}}

{

\newcommand\Bp[1][\Y]{{T^\CT_\Prog}}

\newcommand\cexists[1]{\exists_{\,\text{-}{#1}}}

\newcommand\prth[1]{{\left( {#1} \right)}}

\newcommand\esum[2]{\ifthenelse{\equal{#2}{}}{[#1]}{[#1|#2]}}

\renewcommand\geq\geqslant
\renewcommand\subset{\subseteq}

\newcommand\la{\leftarrow}

\newcommand{\ras}{\twoheadrightarrow}
\newcommand{\rat}{\ra^+}
\newcommand{\rae}{\ra^\equiv}
\newcommand{\las}{\twoheadleftarrow}
\newcommand{\lae}{\la^\equiv}

\def\ra{\rightarrow}

\newcommand\eq{\,\textup{\tt =}\,}

\def\R{{\mathcal R}}

\begin{document}

\submitted{25 March 2012}
\revised{11 June 2012}
\accepted{18 June 2012}

\newcommand\thankstxt{The research leading to these results has
  received funding from the Programme for Attracting Talent  /
  young PHD of the MONTEGANCEDO Campus of International Excellence (PICD),
  the Madrid Regional Government under the CM project P2009/TIC/1465
  (PROMETIDOS), and the Spanish Ministry of Science under the MEC
  project TIN-2008-05624 (DOVES).}

\title[Diagrammatic confluence for {C}onstraint~{H}andling~{R}ules]
{Diagrammatic confluence for {C}onstraint~{H}andling~{R}ules\textsuperscript{\thanks{\thankstxt}}}
\author[R\'emy Haemmerl\'e]{R\'emy Haemmerl\'e\\
Technical University of Madrid}

\maketitle

\begin{abstract}

  Confluence is a fundamental property of {C}onstraint~{H}andling~{R}ules
  ({CHR}) since, as in other rewriting formalisms, it guarantees that the
  computations are not dependent on rule application order, and also
  because it implies the logical consistency of the program
  declarative view.
  In this paper we are concerned with proving the confluence of
  non-terminating CHR programs.
  For this purpose, we derive from van Oostrom's decreasing diagrams
  method a novel criterion on CHR critical pairs that generalizes all
  preexisting criteria.%
  We subsequently improve on a result on the modularity of CHR
    confluence, which permits modular combinations of possibly
    non-terminating confluent programs, without loss of confluence.


\end{abstract}

\newcommand\mycite[2]{\citeANP{#1}#2~\citeyear{#1}}

\section{Introduction}
\label{section:introduction}

{C}onstraint~{H}andling~{R}ules ({CHR}) is a committed-choice constraint
logic programming language, introduced by
\citeNS{fru_chr_overview_jlp98} for the easy development of constraint
solvers. It has matured into a general-purpose concurrent programming
language. Operationally, a {CHR} program consists of a set of guarded
rules that rewrite multisets of constrained atoms.
Declarative\-ly, a {CHR} program can be viewed as a set of logical
implications executed on a deduction principle.

\medskip 

Confluence is a basic property of rewriting systems{}. 
It refers to the fact that any two finite computations starting from a common
state can be prolonged so as to eventually meet in a common state again.
Confluence is an important property for any rule-based language,
because it is desirable for computations to not be dependent on 
a particular rule application order. 
In the particular case of {CHR}, this property is even more desirable,
as it guarantees the correctness of a
program~\cite{AFM99constraints,HaemmerleGH11ppdp}: any program
confluent  has a consistent
logical reading.
Confluence of a {CHR} program is also a fundamental prerequisite for 
logical completeness results \cite{AFM99constraints,Haemmerle11tplpA},
makes possible program
parallelization~\cite{fru_parallel_union_find_iclp05,meister_preflow_push_wlp06},
and may simplify program equivalence
analyses~\cite{abd_fru_equivalence_cp99,Haemmerle11tplpB}.

\medskip

Following the pioneering research of \citeNS{AFM96cp}, most 
existing work dealing with the confluence of CHR limits itself to
terminating programs (see for instance the works by
\citeNS{Abdennadher97cp} and \citeNS{DSS07iclp}). %
Nonetheless, proving confluence without global termination
 assumptions
is still a worthwhile
objective.

From a theoretical point of view, this is an interesting topic,
because, as illustrated by the following example typical CHR programs
fail to terminate on the level of abstract semantics, even if
they do terminate on more concrete levels. %
Indeed, number of analytical results for the language rest on the
notion of confluence, but only when programs are considered with
respect to abstract semantics. %
For instance, in the current state of knowledge, even a result as
important as the guarantee of correction by confluence only holds when
programs are considered with respect to the most general operation
semantics for CHR, namely the {\em very abstract semantics}.

\newcommand\mleq[2]{{#1} \leq {#2}}%

\begin{example}[Partial order constraint]
\xdef\labelExampleLeq{\thetheorem}
\label{example:leq}
\renewcommand\Prog{\mathcal P_{\labelExampleLeq}}%
Let $\Prog$ be the classic CHR introductory example, namely the
constraint solver for partial order. This consists of the following
four rules, which define the meaning of the {\em user-defined} symbol
$\mathord{\mleq{}{}}$ using the built-in equality constraint $\eq{}{}$: 
\[
  \begin{array}{lcl}
 \textit{duplicate} &\at& \mleq{x}{y} \,\backslash \, \mleq{x}{y} \simplif \top \\
 \textit{reflexivity} &\at& \mleq{x}{x} \simplif \top \\
 \textit{antisymmetry} &\at& \mleq{x}{y}, \mleq{y}{x} \simplif x \eq y \\
 \textit{transitivity} &\at& \mleq{x}{y}, \mleq{y}{z} \propag\mleq{x}{z} \\
  \end{array}
\]

The $\textit{duplicate}$ rule implements so-called {\em duplicate
 removal}. In other words, it states that if two copies of the same
user-defined atom are present, then one of them can be removed. The
\textit{reflexivity} and \textit{transitivity} rules respectively state
that any atom of the form $\mleq{x}{x}$ can be removed, and that two
atoms $\mleq{x}{y}$ and $\mleq{x}{y}$ can be substituted with the
built-in constraint $x \eq y$. Finally, the \textit{transitivity} rule
is a {\em propagation rule}. It states that if 
$\mleq{x}{y}$ and $\mleq{y}{z}$ are present, then the atom
$\mleq{x}{z}$ may be added.

It is well know that this program, like any other program using
propagation rules, faces the so-called {\em trivial non-termination
  problem} when considered with respect to the very abstract
semantics. Indeed, for these semantics, a propagation rule applies to
any state it produces, leading to trivial loops. %
In order to solve this problem, \citeNS{Abdennadher97cp} proposed a
token-based semantics in which propagation rules may be applied only
once to the same combination of atoms. %
Nonetheless, such a proposal does not solve all the problems of
termination. Indeed 
the \textit{transitivity} rule may loop on queries containing a cycle in a
chain of inequalities when considered against
\citeANP{Abdennadher97cp}'s semantics. Consider, for instance, the
query $\mleq{x}{y}, \mleq{y}{x}$.

In fact, in order for $\Prog$ to be terminating, the rules of \textit{reflexivity},
\textit{antisymmetry}, and \textit{transitivity} must have
priority over the \textit{transitivity} rule. This behaviour can be
achieved by considering concrete semantics, such as the refined semantics
of \linebreak \citeNS{duck_stuck_garc_holz_refined_op_sem_iclp04}. These
semantics reduce the non-determinism of the CHR execution model by
applying the rules in textual order.

In exchange for gaining termination, the most concrete semantics lose
a number of analytical results. For instance, as explained by
\citeNS{Fruehwirth09cambridge}, although any CHR program can be run in
parallel in abstract semantics, one can obtain incorrect results for
programs written with the refined semantics in mind. Indeed, if the
result of a program relies on a particular rule application order,
parallel execution will garble this order, leading to unexpected
results. Interestingly, confluence on an abstract (but possibly
non-terminating) level may come to the rescue of the most concrete
semantics: If a program is confluent on a semantic level where the
rule application order is not specified, then the result will not be
dependent on the particular application
order. 
Similar considerations have been discussed for equivalences of CHR
programs~\cite{Haemmerle11tplpB}.

 \end{example}

 From a more practical point of view, proving confluence without
 the assumption of termination is important, because 
 it may be desirable to prove the confluence of a program for which termination cannot
 be inferred. Indeed, there exist very simple programs, such as the Collatz
 function, for which termination is only a conjecture~\cite{guy2004unsolved}. Furthermore,
 since CHR is now a general-purpose language, analytical tools for the
 language must handle programs that do not terminate on any semantic
 level---for instance, interpreters for a Turing-complete
 language~\cite{sney_schr_demoen_chr_complexity_toplas09}, or 
 typical concurrent programs (see the numerous examples of concurrent
 systems given by \citeNS{Milner99cambridge}). %
 We have also recently demonstrated
 that non-terminating execution models for CHR yield elegant
 frameworks for programming with coinductive
 reasoning~\cite{Haemmerle11tplpA}. %
 As a motivating example for the class of intrinsically
 non-terminating programs, we will use the following solution for the
 seminal dining philosophers problem. %

\begin{example}[Dining philosophers]
\newcommand\frk{\textbf{frk}}%
\newcommand\eat{\textbf{eat}}%
\newcommand\thk{\textbf{thk}}%
\renewcommand\frk{\textbf{f}}%
\renewcommand\eat{\textbf{e}}%
\renewcommand\thk{\textbf{t}}%
\xdef\labelExamplePhilos{\thetheorem}%
\label{example:chr:philos}%
\renewcommand\Prog{\mathcal P_{\labelExamplePhilos}}%
Consider the following {CHR} program $\Prog$ that implements a 
solution to the dining philosophers problem extended to count the number of times a
philosopher eats:
\[
\begin{array}{lllll}
 eat  &\at&  \thk(x, y, i), \frk(x), \frk(y) &\simplif& \eat(x, y, i+1) \\
 thk  &\at&  \eat(x, y, i) &\simplif& \frk(x), \frk(y), \thk(x, y, i)
\end{array}
\]
The atom $\frk(x)$ represents the fork $x$, the atom $\eat(x, y, i)$
(resp.\ $\thk(x, y, i)$) represents an eating (thinking) philosopher
seated between forks $x$ and $y$, who has already eaten $i$ times. On
the one hand, the rule $eat$, states that if a thinking philosopher is
seated between two forks lying on the table, then he may start eating
once he has picked up both forks. On the other hand, the rule $thk$
states that a philosopher may stop eating if he puts down the forks he
has been using. %
The initial state corresponding to $n$ dining philosophers seated
around a table can be encoded by the set of atoms $\frk(1), \thk(1, 2,
0), \frk(2), \thk(2, 3, 0), \cdots \frk(n), \thk(n, 1, 0)$.

Despite the fact that this program is intrinsically non-terminating, we may be
interested in its confluence, for example, so that we may make use of one of 
the previously mentioned applications (e.g. confluence simplifies observational
equivalence~\cite{Haemmerle11tplpB}). Confluence of $\Prog$ may also
simplify the proofs of fundamental properties of concurrent systems,
such as, for instance, the absence of deadlock: Starting from the
initial state, one can easily construct a derivation where the
$i^\text{th}$ philosopher ($i \in 1,\dots, n$) has eaten an arbitrary
number of times. Hence if $\Prog$ is confluent, we can then infer that it is
possible to extend any finite derivation such that the $i^\text{th}$
philosopher eats strictly more, \ie\ no derivation leads to a
deadlock.
\end{example}

\smallskip

To the best of our knowledge, the only existing principle for proving
confluence of non-terminating programs is the so-called {\em strong
  confluence} criterion
\cite{HF07rta,raiser_tacchella_confluence_non_terminating_chr07}.
However this criterion appears to be too weak to apply to common CHR
programs, such as \cref{example:leq,example:chr:philos}. In this paper, we are concerned with
extending CHR confluence theory to be able to capture a large class of
possibly non-terminating programs.
For this purpose we derive from the so-called {\em decreasing
  diagrams} technique a novel criterion that generalizes all existing
confluence criteria for {CHR}.
The decreasing diagrams technique is a method developed by
\citeN{diagrams} which subsumes all sufficient conditions for
confluence. Applying this method requires that all local rewrite peaks
(\ie\ points where the rewriting relation diverges because of
non-determinism) can be completed into so-called decreasing diagrams.

\medskip 

The present paper presents two main contributions.  In
\cref{section:confluence}, we present a particular instantiation of
the decreasing diagrams technique to {CHR}, and show that in the
context of this particular instantiation, the verification of
decreasingness can be restricted to the standard notion of critical
pairs.  Then in \cref{section:modularity}, we extend the so-called
{\em modularity of confluence} \cite{Fruehwirth09cambridge} so as to
be able to combine programs which have independently been proven
confluent, without losing confluence.

\section{Preliminaries on abstract confluence}
\label{section:abstract}

In this section, we gather some required notations, definitions, and results
on the confluence of abstract rewriting systems.
\citeANP{Terese03}'s compendium
\citeNN{Terese03} can be referred to for a more detailed presentation.

\medskip

A {\em rewrite relation} (or {\em rewrite} for short) is a binary
relation on a set of objects $E$.
For any rewrite $\mathord{\ra}$, the symbol $\mathord{\la}$ will denote its
converse, $\mathord{\rae}$ its reflexive closure, %
$\mathord{\rat}$ its transitive closure, %
and $\mathord{\ras}$ its transitive-reflexive closure. %
We will use $\mathord{\ra_\alpha} \cdot \mathord{\ra_\beta}$ to
denote the left-composition of all rewrites $\mathord{\ra_\alpha}$ and
$\mathord{\ra_\beta}$.  A {\em family} of rewrites is a set
$(\mathord{\ra_\alpha})_{\alpha \in I}$ of rewrites indexed by a set $I$ of
labels. %
For such a family and any set $K$, $\mathord{\ra_K}$ will denote the union
$\bigcup _{\alpha \in (K \cap I)} \prth{\mathord{\ra_\alpha}}$.

 %
A {\em
    reduction} is a finite sequence of rewriting steps of the form
  $(e_0 \rap_{\alpha_1}e_1 \rap_{\alpha_2} \cdots \rap_{\alpha_n}
  e_n)$. Such a reduction would be abreviated as $e_0 \rap_{\vec
    \alpha} e_n$ with $\vec \alpha = \alpha_1,\alpha_2, \dots
  ,\alpha_n$ when the intermediary states $e_1, \dots, e_{n-1}$ are not relevant.
  A {\em peak} is a pair of reductions $e_l \lap_{\vec \alpha} e
  \rap_{\vec \beta} e_r$ from a common element $e$.  A {\em
    local peak} is a peak formed by two one-step reductions. A {\em
    valley} is a pair of reductions $e_l\rap_{\vec \alpha} e'
  \lap_{\vec \beta} e_r$ ending in a common element $e'$.
  A peak $e_l \lap_{\vec \alpha} e \rap_{\vec \alpha'} e_r$ \linebreak
  is {\em joinable} by $\mathord{ \rap_{\beta}} \cdot
  \mathord{\lap_{\beta'}}$ if it is true that $e_l \rap_{\beta} \cdot \lap_{\beta'}
  e_r$.

 A rewrite $\mathord{\ra}$ is {\em terminating} if there is no infinite sequence
 of the form $e_0 \ra e_1 \ra e_2 \dots$
Furthermore, we will say that $\mathord{\ra}$ is {\em confluent} if 
$(\mathord{\las} \cdot
\mathord{\ras}) \subset (\mathord{\ras} \cdot \mathord{\las})$ holds, {\em locally confluent} if $(\mathord{\la}
\cdot \mathord{\ra}) \subset (\mathord{\ras} \cdot \mathord{\las})$ holds, and {\em strongly
  confluent\footnote{For the sake of simplicity, we use a definition
    weaker than the one of \citeN{Huet80jacm}. %
    It is worth noting, that the counterexamples given in introduction
    stay relevant for the general definition.  } } %
if $(\mathord{\la} \cdot \mathord{\ra}) \subset (\mathord{\rae} \cdot
\mathord{\lae})$ holds.
\cref{fig:arr:confluence,fig:arr:local,fig:arr:strong} graphically
represent these definitions.  Following standard diagrammatic
notation, solid edges stand for universally quantified rewrites,
while dashed edges represent existentially quantified rewrites.

\begin{figure}
\figrule
\begin{minipage}[b]{0.26\linewidth}\centering
\begin{tikzpicture}[scale=1]
  \node[inner sep=3] (e) at (0,1) {};
  \node[inner sep=3](e1) at (-1,0) {};
  \node[inner sep=3](e2) at (1,0) {};
  \node[inner sep=3](e') at (0,-1) {};
 \draw[->>,>=angle 90] (e) to 
  (e1);
\draw[->>,>=angle 90] (e) to  
(e2);

\draw[dashed,->>,>=angle 90](e1) to
(e'); 
\draw[dashed,->>,>=angle 90](e2) to 
(e');
\end{tikzpicture}   
   \caption{Confluence}
\label{fig:arr:confluence}
\end{minipage}
\hfill
\begin{minipage}[b]{0.33\linewidth}
\centering
\begin{tikzpicture}[scale=1]
  \node[inner sep=1.5] (e) at (0,1) {};
  \node[inner sep=1.5](e1) at (-1,0) {};
  \node[inner sep=1.5](e2) at (1,0) {};
  \node[inner sep=3](e') at (0,-1) {};

\draw[->,>=angle 90] (e) to 
  (e1);
\draw[->,>=angle 90] (e) to  
(e2);

\draw[dashed,->>,>=angle 90](e1) to
(e'); 
\draw[dashed,->>,>=angle 90](e2) to
(e'); 
\end{tikzpicture}   
   \caption{Local Confluence}
\label{fig:arr:local}  
\end{minipage}
\hfill
\begin{minipage}[b]{0.34\linewidth}
\centering
\begin{tikzpicture}[scale=1]
  \node[inner sep=2] (e) at (0,1) {};
  \node[inner sep=2](e1) at (-1,0) {};
  \node[inner sep=2](e2) at (1,0) {};
  \node[inner sep=2](e') at (0,-1) {};

\draw[->,>=angle 90] (e) to 
  (e1);
\draw[->,>=angle 90] (e) to  
(e2);

\draw[dashed,->,>=angle 90](e1) to
   node [inner sep=3, sloped, above] {$\equiv$} 
(e'); 
\draw[dashed,->,>=angle 90](e2) to
   node [inner sep=3, sloped, above] {$\equiv$} 
(e'); 
\end{tikzpicture}   
   \caption{Strong confluence}
\label{fig:arr:strong}
\end{minipage}
\figrule
 \end{figure}

 By the seminal lemma of \citeN{Newman42am}, we know that a
  terminating and locally confluent rewrite is confluent. Another
  famous result due to \citeN{Huet80jacm} ensures that strong
  confluence implies confluence.  

\renewcommand\succeq{\succcurlyeq}
\newcommand\succeqR[1][]{\,\text{\rotatebox[origin=c]{270}{$\succeq$}}\!_{#1}\,}
\newcommand\succR[1][]{\,\text{\rotatebox[origin=c]{270}{$\succ$}}\!_{#1}\,}

\medskip

We now present a slight variation due to \citeN{HM10ijcar} of the
so-called decreasing diagrams technique, which is more suitable for
our purposes.
The interest of the decreasing diagrams method~\cite{diagrams} is that
it reduces problems of general confluence to problems of local
confluence. In exchange, the method requires the {\em confluence
  diagrams} (\ie\ the way peaks close) to be decreasing with respect
to a labeling provided with a wellfounded preorder. The method is
complete in the sense that any countable confluence rewrite can be
equipped with such a labeling. But because confluence is an
undecidable property, finding such labeling may be difficult.

In the rest of this paper, we will say
that a preorder $\succeq$ is {\em wellfounded}, if the strict
preorder $\succ$ associated with $\succeq$ (\ie\ $\alpha \succ \beta$
iff $\alpha \succeq \beta$ but not $\beta \succeq \alpha$) is
a terminating relation.
\label{defn:decreasing}
Let $(\ra_\alpha)_{\alpha \in I}$ be a family of rewrites and
$\succeq$ be a wellfounded preorder on $I$. A local peak $e_l
\la_\alpha e \ra_\beta e_r$  ($\alpha, \beta \in I$)  is {\em decreasing} with respect to
$\succeq$ if the following holds:
\begin{equation*}
\tag{$\star$}
   e_l \ras_{\succR\{\alpha\}} \cdot \rae_{\succeqR\{\beta\}}\cdot
\ras_{\succR\{\alpha, \beta\}} e' \las_{\succR\{\alpha, \beta\}}
\cdot\lae_{\succeqR\{\alpha\}}\cdot \las_{\succR\{\beta\}} e_r %
\end{equation*}
where for any set $K$ of labels, $\succeqR K$ stands for
$\left\{\gamma \in I \mid \exists \delta \in K . \delta \succeq \gamma
\right\}$ and $\succR K$ for $\left\{\gamma \in I \mid \exists \delta
  \in K . \delta \succ \gamma \right\}$.  A family
$(\ra_\alpha)_{\alpha \in I}$ of rewrites is {\em (locally)
  decreasing} if all local peaks of the form $u \la_\alpha \cdot
\ra_\beta v$ ($\alpha, \beta \in I$) are decreasing with respect to a
common wellfounded preorder on $I$.  A rewrite is {\em (locally)
  decreasing} if it is the union of some decreasing families of
rewrites.  Property $(\star)$ is graphically represented in
\cref{fig:diagram:decreasingness}.
 
\begin{figure}
\figrule
\centering
\begin{tikzpicture}[xscale=2,yscale=1.9]
  \node[inner sep=1.5] (e) at (0,1) {};
  \node[inner sep=1.5](e1) at (-1,0) {};
  \node[inner sep=1.5](e2) at (1,0) {};
  \node[inner sep=3](e') at (0,-1) {};
  \node[inner sep=1.5](e11) at ($(e1)!0.333!(e')$) {};
  \node[inner sep=1.5](e12) at ($(e1)!0.666!(e')$) {};
  \node[inner sep=1.5](e21) at ($(e2)!0.333!(e')$) {};
  \node[inner sep=1.5](e22) at ($(e2)!0.666!(e')$) {}; 

\draw[->,>=angle 90] (e) to 
  node [sloped, below] {${}_\alpha$}
  (e1);
\draw[->,>=angle 90] (e) to  
   node [sloped, below] {${}_\beta$}
(e2);

\draw[dashed,->>,>=angle 90](e1) to
   node [sloped, below] {${}_{\succR\{\alpha\}}$}
(e11); 
\draw[dashed,->,>=angle 90](e11) to 
   node [inner sep=3, sloped, above] {$\equiv$} 
   node [sloped, below] {${}_{\succeqR\{\beta\}}$} (e12);
\draw[dashed,->>,>=angle 90](e12) to 
 node [sloped, below] {${}_{\succR\{\alpha,\beta\}}$}
(e'); 

\draw[dashed,->>,>=angle 90](e2) to 
node [sloped, below] {${}_{\succR\{\beta\}}$}
 (e21);
\draw[dashed,->,>=angle 90](e21) to 
 node [inner sep=3, sloped, above] {$\equiv$} 
 node [sloped, below] {${}_{\succeqR\{\alpha\}}$}
 (e22);
\draw[dashed,->>,>=angle 90](e22) to 
  node [sloped, below] {${}_{\succR\{\alpha,\beta\}}$}
 (e');
\end{tikzpicture}   
   \caption{Local decreasingness}
\label{fig:diagram:decreasingness}  
\figrule
 \end{figure}
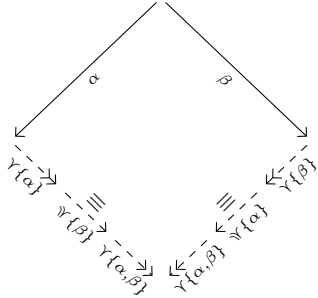

\begin{theorem}[Decreasing Diagram \cite{diagrams}] 
\label{thm:decreasing}
A countable rewrite  is confluent if and only if it is locally decreasing.
\end{theorem}

\medskip

We recall now some other state-of-the-art results 
which
will be used later.

\begin{lemma}[\citeNP{Terese03}]
\label{lemma:terese}
\begin{enumerate}[{\it (i)}]
\item 
  For all rewrites $\mathord{\ra_1}$, $\mathord{\ra_2}$ if
  $(\la_1\cdot\ras_2)\subset (\ras_2\cdot\las_1)$, then
  $(\las_1\cdot\ras_2)\subset (\ras_2\cdot\las_1)$.

\item For all rewrites $\mathord{\ra_1}$,$\mathord{\ra_2}$ s.t.\ 
  $\mathord{\ra_1} \subset \mathord{\ra_2} \subset \mathord{\ras_1}$,
    $\mathord{\ra_2}$ is confluent iff $\mathord{\ra_1}$ is confluent.
\end{enumerate}
\end{lemma}

\section{Preliminaries on {C}onstraint~{H}andling~{R}ules}
\label{section:preliminaries}

In this section, we recall the syntax and the semantics of {CHR}.
\citeANP{Fruehwirth09cambridge}'s book
\citeyear{Fruehwirth09cambridge} can be referred to for a more general
overview of the language.

\subsection{Syntax}

The formalization of {CHR} assumes a language of \emph{(built-in)
  constraints} containing equality over some theory $\CT$, and defines
\emph{(user-defined) atoms} using a different set of predicate
symbols. %
In the following, $\R$ will denote an arbitrary set of identifiers.
By a slight abuse of notation, we allow confusion of conjunctions and
multiset unions, omit braces around multisets, and use the comma for
multiset union. %
We use $\fv(\phi)$ to denote the set of free variables of a formula
$\phi$. %
The notation $\cexists{\psi} \phi$ denotes the existential closure of
$\phi$ with the exception of free variables of $\psi$.

\medskip

A {\em ({CHR}) program} is a finite set of eponymous rules of the form:
 \[( r
\at \KK \backslash \HH \simplif \GG \mid \BB; \CC )\]
where $\KK$ (the {\em kept head}), $\HH$ (the {\em removed head}), and
$\BB$ (the {\em user body}) are multisets of atoms, $\GG$ (the {\em
  guard}) and $\CC$ (the {\em built-in body}) are conjunctions of
constraints and, $r\in \R$ (the {\em rule name}) is an identifier
assumed unique in the program.
Rules in which both heads are empty are prohibited.  An empty guard
$\top$ (resp. an empty kept head) can be omitted with
the symbol~$\mid$ (resp.\@ with the symbol~$\backslash$). %
Rules are divided into two classes: {\em simplification
  rules}\footnote{Unlike standard presentations, our definition does not distinguish between simplification rules form the so-called simpagation rules.}  if the
removed head is non-empty and {\em propagation rules} otherwise.
Propagation rules can be written using the alternative syntax: \[( r
\at \KK \propag \GG \mid \BB; \CC{} ) \]

\subsection{Operational semantics}
\label{subsection:chr:operational}

In this section, we recall the equivalence-based operational semantics
$\omega_e$ of \citeN{RBF09chr}. It is equivalent to the {\em very
  abstract} semantics $\omega_{va}$ of \citeN{fru_chr_overview_jlp98},
which is the most general operational semantics of {CHR}. We prefer
the former because it includes an rigorous notion of equivalence,
which is an essential component of confluence analysis.

\medskip 

A {\em ({CHR}) state} is a tuple $\state{\CC}{\EE}{\bar x}$, where
$\CC$ (the {\em user store}) is a multiset of atoms, $\EE$ (the {\em
  built-in store}) is a conjunction of constraints, and $\vec x$ (the
{\em global variables}) is a finite set of variables. %
Unsurprisingly, the {\em local variables} of a state are those
variables of the state which are not global. %
  When no confusion can occur, we will syntactically merge user and
  built-in stores. We may futhermore omit the global variables
  component when states have no local variables.
%
 %
In the following, we use $\Sigma$ to denote the set of states. %
Following \citeANP{RBF09chr}, we will always implicitly
consider states modulo a structural equivalence. Formally, this {\em
  state equivalence} is the least equivalence relation $\equiv$ over
states satisfying the following rules:
\label{definition:equivalence}
   \begin{itemize}
 \item
  \label{preliminaries:abstract:equivalence:transformation}
 $\state{\EE}{\CC}{\vec x} \equiv \state{\EE}{\DD}{\vec x}$ if 
$\CT \vDash \cexists{\prth{\EE,\vec x}}{
    \CC} \leftrightarrow \cexists{\prth{\EE,\vec x}}{\DD}$
  \item \label{preliminaries:abstract:equivalence:inconsistent}
$\state{\EE}{\bot}{\vec x}\equiv \state{\FF}{\bot}{\vec y}$ 
  \item \label{preliminaries:abstract:equivalence:equality} $
     \state{\AA, c}{\CC, c\mathord{=}
       d}{\vec x} \equiv \state{\AA, d}{\CC,  c \mathord{=} d}{\vec x}$
 \item \label{preliminaries:abstract:equivalence:omission}
$\state{\AA}{\CC}{ \vec x} \equiv
  \state{\AA}{\CC}{\{y\} \cup \vec x}$ if  $y \notin \fv(\AA, \CC)$.
\end{itemize}

Once states are considered modulo equivalence, the operation semantics
of {CHR} can be expressed by a single rule.  Formally the 
operational semantcs 
of a program $\Prog$ is given by the least relation $\rap$ on
states satisfying the rule:
\begin{gather*}
\frac{
\prth{r \at \KK \backslash \HH \simplif \GG | \BB; \CC}  \in \Prog\rho
 \quad \lv(r) \cap \fv(\EE, \DD, \vec x) = \emptyset
}{ %
\state{\KK, \HH, \EE}{ \GG,  \DD}{ \vec x}
\rap
\state{\KK, \BB, \EE}{\GG, \CC, \DD}{\vec x}
}
\end{gather*}
where $\rho$ is a renaming. A program $\Prog$ is {\em confluent} (resp.\@
{\em terminating}) if $\rap$ is confluent (resp.\@ terminating).

\medskip 

Before going further, we recall an important property of {CHR}
semantics. This property, monotonicity, means that if a transition is
possible in a state, then the same transition is possible in
any larger state. To help reduce the level of verbosity
we introduce the notion of the {\em quantified
  conjunction} of states~\cite{HF07rta}.
This operator allows the composition of states with disjoint local
  variables while quantifying some of their global variables (\ie\
  changing global variables into local ones).  Formally,
  the quantified conjunction is a binary operator on states
  parametrized by a set of variables $\vec z$ satisfying:
\[
\state{\EE}{\CC}{\vec x} \oplus_{\vec z} \state{\FF}{\DD}{\vec y} =
\state{\EE,\FF}{\CC,\DD}{(\vec x\vec y) \setminus \vec z}
 \text{ if }  (\fv(\EE,\CC) \cap \fv(\FF,\DD)) \subset (\vec x \cap \vec y) 
\] 
Note the side condition is not restrictive, as local variables can
always be renamed using the implicit state equivalence.

\begin{proposition}[Monotonicity of {CHR}]
  Let $\Prog$ be a {CHR} program, $S$, $S_1$, $S_2$ be {CHR} states, and
  $\vec x$ be a set of variables.
\[
\text{ If } S_1 \rap S_2 \text{, then } S_1 \oplus_{\vec x} S \rap S_2
\oplus_{\vec x} S
\]
\end{proposition}

\subsection{Declarative semantics}
\label{subsection:preliminaries:declarative}

Owing to its origins in the tradition of CLP, the {CHR} language
features declarative semantics through direct interprestation in
first-order logic.  Formally, the {\em logical reading} of a rule of
the form:
\[ 
   \KK \backslash \HH \simplif \GG \mid  \BB;\CC
\]
is the guarded equivalence:
\[
  \forall \prth{\prth{\KK \wedge \GG} \rightarrow \prth{\HH
      \leftrightarrow \cexists{\prth{\KK, \HH}}\prth{\GG \wedge \CC
        \wedge \BB}}}
\]
The {\em logical reading} of a program $\Prog$ within a theory $\CT$ is
the conjunction of the logical readings of its rules with the
constraint theory $\CT$. It is denoted by $\CP$. 

\medskip

Operational semantics is sound and complete with respect to this
declarative semantics~\cite{fru_chr_overview_jlp98,AFM99constraints}.
Furthermore, any program confluent with respect to $\omega_e$
 has a consistent logical
reading~\cite{AFM99constraints,HaemmerleGH11ppdp}.

\section{Diagrammatic confluence for {C}onstraint {H}andling {R}ules} 
\label{section:confluence}

In this section, we are concerned with proving confluence of a large
class of {CHR} programs. Indeed, as explained in the introduction,
existing criteria are not sufficiently powerful to infer confluence of
common non-terminating programs.  (See
\cref{example:diagrams:leq,example:diagrams:philos} for 
concrete examples). To avoid this limitation, we will derive from the
decreasing diagrams technique a novel csriterion on {CHR} critical
pairs that generalizes both local and strong confluence criteria.
 An analogue
 criterion has been developed for linear {T}erm {R}ewriting {S}ystems (TRS)
 \cite{JvO09icalp}.

\subsection{Labels for {C}onstraint {H}andling {R}ules}

In order to apply the decreasing diagram technique to {CHR}, we will
need first to label {CHR} transitions. In this work, we will use two
labelings proposed by \citeN{vanOostrom08rta} for 
{TRS}.  The first one is the so-called {\em
  rule-labeling}. 
It consists of labeling each transition $a \rap b$ with the name of
the applied rule. %
This labeling is ideal for capturing strong confluence-like properties
for linear TRS.  Within the proof of our main result, we will also use
the so-called {\em self-labeling} which consists of labeling each
transition $a \rap b$ with its source $a$. This second labeling
captures the confluence of terminating rewrites.

\medskip

In practice, we will assume that the set $\R$ of rule identifiers is
defined as a disjoint union $\R_i\uplus \R_c$. For a given program
$\Prog$, we denote by $\Prog^i$ (resp.\ $\Prog^c$) the set of rules form $\Prog$ built with
$\R_i$ (resp. $\R_c$). %
We call $\Prog^i$ the {\em inductive} part of $\Prog$, because we will
subsequently assume that $\Prog^i$ is terminating, while $\Prog^c$ will be
called {\em coinductive}, as it will be
typically  non-terminating.

\begin{definition}[Rule-labeling]
\label{def:chr:labels}
The {\em rule-labeling} of a {CHR} program $\Prog$ is the family of
rewrites $(\mathord{\rap_{r}})_{r\in \R}$ indexed by rule identifiers,
where $\mathord{\rap_{r}} = \mathord{\rap[\{r\}]}$.  A preorder
$\succeq$ on rule identifiers is {\em admissible}, if any inductive
rule identifier is strictly smaller than any coinductive one (\ie\
for any $r_i \in \R_i$ and any $r_c \in  \R_c$,  $r_c \succ r_i$ holds).
\end{definition}

\subsection{Critical peaks}

In TRS, the basic techniques used to prove confluence consist of
showing various confluence criteria on a finite set of special cases,
called {\em critical pairs}. 
Critical pairs are
generated by a superposition algorithm, in which one attempts to
capture the most general way the left-hand sides of the two rules of
the system may overlap. The notion of critical pairs has been successfully
adapted to {CHR} by \citeN{AFM96cp}. Here, we introduce a slight
extension of the notion that takes into account the rule-labeling we
have just defined.

\begin{definition}[Critical peak]
  Let us assume that $r_1$ and $r_2$ are {CHR} rules renamed apart:
\[
\left( r_1 \at \KK_1 \backslash \HH_1 \simplif \GG_1 \mid \BB_1; \CC_1 \right) \in \Prog_1 
\qquad \qquad 
\left( r_2 \at \KK_2 \backslash \HH_2 \simplif \GG_2 \mid \BB_2; \CC_2 \right) \in \Prog_2
\]
A {\em critical ancestor (state)} $S_c$ for the rules $r_1$ and $r_2$ is
a state of the form:
\[
S_c = \state{\HH^\Delta_1,\HH_1^\cap, \HH_2^\Delta}{\DD}{\vec x}
\]
 satisfying the following properties:
\begin{itemize}
 \item $(\KK_1, \HH_1) \dot = (\HH_1^\Delta, \HH_1^\cap)$, 
       $(\KK_2, \HH_2) \dot = (\HH_2^\Delta, \HH_2^\cap)$,  
       $\HH_1^\cap  \not = \emptyset$, and $\HH_2^\cap \not = \emptyset$;
     \item  $\vec {x_1} =
       \fv(\KK_1, \HH_1)$, $\vec {x_2} = \fv(\KK_2, \HH_2)$ and $\vec x = \vec {x_1} \cup \vec {x_2}$;
\item $\DD = (\HH_1^\cap\mathord{\dot =}  \HH_2^\cap, \GG_1, \GG_2)$ and 
 $\exists \DD$ 
is $\CT$-satisfiable;
\item $\HH^\cap_1 \not \subset \KK_1$ or $\HH^\cap_2 \not \subset \KK_2$.
 \end{itemize}
 Then the following tuple is called a {\em critical peak} between
 $r_1$ and $r_2$ at $S_c$:
\[
\state{\KK_1, \BB_1, \HH_2^\Delta}{\DD, \CC_1}{\vec x}
 \lap_{r_1} S_c \rap_{r_2} \state{\KK_2, \BB_2,
   \HH_1^\Delta}{\DD, \CC_2}{\vec x}
\]
A critical peak between a program $\Prog$ and a program $\Q$ is a
critical peak between a rule of $\Prog$ and a rule of $\Q$. A critical
peak of a program $\Prog$ is a critical peak between $\Prog$ and itself. %
 A critical peak is {\em inductive} if it involves only inductive rules
(\ie\ a critical peak of $\Prog^i$), or {\em coinductive} if it involves
at least one coinductive rule (\ie\ a critical peak between $\Prog^c$
and $\Prog$).
\end{definition}

\begin{example}
\label{example:peak:leq}
\renewcommand\Prog{\mathcal P_{\labelExampleLeq}}%
Consider the solver partial order $\Prog$, given in \cref{example:leq}.
The following ciritial peak stems from overlapping the heads of the
rules $\textit{antisymmetry}$ and $\textit{transitivity}$:
\[
\sstate{x\eq y} \lap_\textit{anti.} \sstate{\mleq{x}{y}, \mleq{y}{x}} \rap_\textit{trans.}
\sstate{\mleq{x}{y}, \mleq{y}{x}, \mleq{x}{x}}
\]
\end{example}

\subsection{Rule-decreasingness}

We now come to our main result, showing that the study of
decreasingness with respect to the rule-labeling can be restricted to critical
peaks without loss of generality.

\begin{definition}[Critical rule-decreasingness]
\label{defn:chr:decreasing}
A program $\Prog$ is {\em (critically) rule-decreasing} \wrt\ 
  an
admissible preorder $\succeq$  if:
\begin{itemize}
\item the inductive part of $\Prog$ is terminating,
\item all inductive critical peaks of $\Prog$ are joinable by $\raps[\Prog^i]\cdot  \laps[\Prog^i]$, and
\item all coinducitve critical peaks of $\Prog$ are decreasing \wrt\ $\succeq$.
\end{itemize}
A program is {\em rule-decreasing} if it is rule-decreasing with
respect to some admissible preorder.  %
A rule-decreasing program is {\em strongly rule-decreasing} if
  it is purely coinductive (\ie\ without inductive rules).
\end{definition}

\begin{theorem}
\label{thm:chr:critical}
Rule-decreasing programs are confluent.
\end{theorem}

\begin{proof*}
  Let us assume that $\Prog$ is a rule-decreasing program  \wrt\ a given
  preorder $\succeq_\R$. 
  Now let $(\rap_\alpha)_{\alpha \in (\Sigma \cup \R_c)}$, the family
  of rewrites indexed by rule or state, be defined as
  
\[
\mathord{\rap_\alpha} = 
\begin{cases}
\rap[\phantom.\Prog^i\phantom.]  \cap \prth{\left\{\alpha\right\} \times \Sigma} & \text{ if } \alpha \in \Sigma   \qquad \qquad \text{(self-labeling on inductive part)} \\
\rap[\{\alpha\}] & \text{ if } \alpha \in \R_c \qquad \quad\;\;  \text{(rule-labeling on coinductive part)} 
\end{cases}
\]
Let $\succeq$ be the union of $\succeq_\R$, $\rapt[\Prog^i]$, and $\{(r,
\alpha) \mid r\in \R \; \& \; \alpha \in I\}$. By assuming without
loss of generality that $\R$ is finite (\ie\ $\succeq_\R$ is trivially
wellfounded), we obtain that $\succeq$ is wellfounded.
  With the help of \cref{thm:decreasing}, it suffices to prove that each
  peak $S_\alpha \lap_\alpha \!S\! \rap_\beta S_\beta$ ($\alpha,
  \beta \in (\R_c \cup \Sigma)$) is decreasing \wrt\ $\succeq$. We
  distinguish two cases:
\begin{enumerate}[1]
\item The rules $r_\alpha$ and $r_\beta$ used to respectively produce
  $S_\alpha$ and $S_\beta$ apply to different parts of $S$.  By
  monotonicity of {CHR} transitions, we infer $S_\alpha
  \rap[\{r_\beta\}] S' \lap[\{r_\alpha\}] S_\beta$.  We have to show
  this {\em valley} respects property $(\star)$ within the definition
  of the decreasing diagrams.  We proceed by cases on the types of the
  rules $r_\alpha$ and $r_\beta$:
  \begin{enumerate}[{1}.1]
  \item $r_\alpha$ is inductive. We have $\alpha=S$, $\alpha \rapt[\Prog^i]
    S_\alpha$, and $S_\beta \rap_{S_\beta} S'$.
\begin{enumerate}[{1.1}.1]
\item $r_\beta$ is inductive. We have $\beta=S$, $\beta \rap[\Prog^i]
  S_\beta$, and $S_\alpha \rap_{S_\alpha} S'$. Since $\Prog^i$ is
  terminating, we infer $\alpha \succ S_\alpha$ and $\beta \succ
  S_\beta$. We conclude $S_\alpha \rap_{S_\alpha} S' \lap_{S_\beta}
  S_\beta$, \ie\ the peak is decreasing \wrt\ $\succeq$.
\item \label{symmetric} $r_\beta$ is coinductive. We have $\beta \in \R_c$, $S_\alpha
  \rap_{\beta} S'$, and $\beta \succ S_\alpha$. We conclude $S_\alpha
  \rap_{S_\alpha} S' \lap_{\beta} S_\beta$, \ie\ the peak is
  decreasing \wrt\ $\succeq$.
  \end{enumerate}
\item $r_\alpha$ is coinductive. We have $\alpha \in \R_c$ and $S_\beta
  \rap_{\alpha} S'$.
\begin{enumerate}[{1.2}.1]
\item $r_\beta$ is inductive. The case is symmetric with case 1.1.2.
\item $r_\beta$ is coinductive. We have $\beta \in \R_c$ and $S_\alpha
  \rap_{\beta} S'$. We conclude $S_\alpha
  \rap_{\alpha} S' \lap_{\beta} S_\beta$, \ie\ the peak is
  decreasing \wrt\ $\succeq$.
  \end{enumerate}
\end{enumerate}

\item The applications of the rules $r_\alpha$ and $r_\beta$ used to respectively
  produce $S_\alpha$ and $S_\beta$ overlap. There should exist a critical peak
  \( R_\alpha \lap_{r_\alpha} S_c \rap_{r_\beta} R_\beta\), a state $R$, and a
  set of variables $\vec y$, such that $S \equiv S_c \oplus_{\vec x}
  R$, $S_\alpha \equiv R_\alpha \oplus_{\vec x} R$, and $R_\beta \equiv R_\beta
  \oplus_{\vec x} R$. We proceed by cases on the types of rules $r_\alpha$ and $r_\beta$:
 \begin{enumerate}[{2}.1]
 \item Both rules are inductive: We have $\beta=\alpha=S$, and by hypothesis we have
\[
R_\alpha \equiv R_\alpha^0\rap[\Prog^i] R_\alpha^1\rap[\Prog^i] \cdots S^m_\alpha \equiv S'  \equiv R^n_\beta \cdots \lap[\Prog^i] R_\beta^1 \lap[\Prog^i] R_\beta^0 
\equiv R_\beta
\]
By monotony of {CHR} we infer:
\[
S_\alpha \equiv S_\alpha^0\rap[\Prog^i] S_\alpha^1\rap[\Prog^i] \cdots S^m_\alpha \equiv S  \equiv S^n_\beta \cdots \lap[\Prog^i] S_\beta^1 \lap[\Prog^i] S_\beta^0 
\equiv S_\beta
\]
where $S_\alpha^i = R_\alpha^i\oplus_{\vec x} R$ (for $i\in 0, \dots
m$), $S_\beta^i = R_\beta^i\oplus_{\vec x} R$ (for $i\in 0, \dots
n$), and $S=S'\oplus R$. By construction of  $(\rap_\alpha)_{\alpha \in \Sigma \times \R_c}$ we get:
\[
S_\alpha \rap_{S_\alpha^0} S_\alpha^1\rap{S_\alpha^1} \cdots S^m_\alpha \equiv S  \equiv S^n_\beta \cdots \lap_{ S_\beta^1}S_\beta^1 \lap_{ S_\beta^0} S_\beta
\]
To conclude about the discussion of the decreasingness of the peak, it is just necessary
to notice that for any $i\in 0, \dots m$ and any $j\in  0, \dots n$,  both
$S \rapt  S_\alpha^i$ and $S \rapt  S_\beta^j$ hold, \ie\ $S_\alpha^i, S_\beta^j \in \succR\{\alpha, \beta\}$.
 \item One of the rules is coinductive.
By hypothesis we have 
\[
R_\alpha
 \ras_{\succR\{r_1\}} \cdot \rae_{\succeqR\{r_2\}}\cdot
\ras_{\succR\{r_1, r_2\}} \cdot \las_{\succR\{r_1, r_2\}}
\cdot\lae_{\succeqR\{r_1\}}\cdot \las_{\succR\{r_2\}} R_\beta
\]
or equivalently by monotony of {CHR}:
\[
\quad
S_\alpha
 \ras_{\succR\{r_1\}} \cdot \rae_{\succeqR\{r_2\}}\cdot
\ras_{\succR\{r_1, r_2\}} \cdot \las_{\succR\{r_1, r_2\}}
\cdot\lae_{\succeqR\{r_1\}}\cdot \las_{\succR\{r_2\}} S_\beta
\quad
\mathproofbox \]
\end{enumerate}
\end{enumerate}

\end{proof*}

\newcommand\diamondTikz[4]{
  \node[inner sep=3] (e) at (0,1) {$S$};
  \node[inner sep=3](e1) at (-1,0) {$S_\alpha$};
  \node[inner sep=3](e2) at (1,0) {$S_\beta$};
  \node[inner sep=3](e') at (0,-1) {$S'$};
 \draw[->,>=angle 90] (e) to 
   node [inner sep=5, below] {$#1$}
  (e1);
\draw[->,dragan] (e) to  
   node [inner sep=5,  below] {$#2$}
(e2);

\draw[dashed,->,dragan](e1) to
   node [inner sep=5,  below] {$#3$} 
(e'); 
\draw[dashed,->,>=angle 90](e2) to 
   node [inner sep=5,  below] {$#4$} 
(e');
}

\cref{thm:chr:critical} strictly subsumes all the criteria for proving
confluence of {CHR} programs we are aware of, namely the local
confluence~\cite{AFM99constraints} and the strong
confluence~\cite{HF07rta} 
criteria.

\begin{corollary}[Local confluence]
\label{corollary:decreasing:local}
A terminating program $\Prog$ is confluent if its critical peaks are
joinable by $\raps[\Prog]\!\cdot \laps[\Prog]$.
\end{corollary}

\begin{corollary}[Strong confluence]
\label{corollary:decreasing:strong}
A  program $\Prog$ is confluent if its critical peaks are joinable by $\rape[\Prog]\!\cdot \lape[\Prog]$.
\end{corollary}

The following examples show that the rule-decreasingness criterion is
more powerful than both local and strong confluence criteria.

\begin{example}
\label{example:diagrams:leq}
\renewcommand\Prog{\mathcal P_{\labelExampleLeq}}%
Consider the solver $\Prog$ for partial order given in
  \cref{example:leq}.
Since $\Prog$ is trivially non-terminating
one cannot apply local confluence criterion. Strong confluence
does not apply either, because of some non-strongly joinable critical
peaks. For instance, considere the peak given at \cref{example:peak:leq}:
\[
\sstate{x\eq y} \lap_\textit{anti.} \sstate{\mleq{x}{y}, \mleq{y}{x}}
\rap_\textit{trans.}  \sstate{\mleq{x}{y}, \mleq{y}{x}, \mleq{x}{x}}
\]
It can be seen that $\sstate{x\eq y}$ may not be reduced, and that
the right-hand side cannot be rewritten into the left-hand side in less
than two steps (e.g.\ by using $\textit{reflexivity}$ and
$\textit{antisymmetry}$ rules).

Nonetheless, confluence of $\Prog$ can be deduced using the full
generality of \cref{thm:chr:critical}. For this purpose, assume that
all rules except $\textit{transitivity}$ are inductive and take any
admissible preorder. Clearly the inductive part of $\Prog$ is
terminating.  Indeed the application of any one of the three first
rules strictly reduces the number of atoms in a state.  %
Then by a systematic analysis of all critical peaks of $\Prog$, we prove that each
  peak can be closed while respecting the
hypothesis of rule-decreasingness.
In fact all critical peaks can be closed without using 
$\textit{transitivity}$. %
   Some rule-decreasing diagrams involving the $\textit{transitivity}$ rule
   are given as examples in \Cref{figure:example:leq}.  %

\begin{figure}
\figrule
\begin{tikzpicture}[xscale=1,yscale=1.2,
forall/.style={->}, 
exists/.style={densely dotted,->}]

\node (Sc) at (0,0) {$\sstate{\mleq{x}{x}, \mleq{x}{y}}$};
\node (S1) at (-1,-1) {$\sstate{\mleq{x}{x}}$};
\node (S2) at (1,-1) {$\qquad \sstate{\mleq{x}{x}, \mleq{x}{y}, \mleq{x}{x}}$};
\node (S2') at (0,-2) {$\sstate{\mleq{x}{y}, \mleq{y}{x}}$};

\draw[forall] (Sc) to node [left] {$anti.$}  (S1);
\draw[forall] (Sc) to node [right] {$\,trans.$}  (S2);
\draw[exists] (S2) to node [right] {$\,reflex.$}  (S2');
\draw[exists] (S2') to node [left] {$anti.$}  (S1);

\node (Sc) at (6,0) {$\sstate{\mleq{x}{y}, \mleq{y}{z}, \mleq{z}{y} }$};
\node (S1) at (5,-1) {$\sstate{\mleq{x}{y}, y \eq z}\qquad $};
\node (S2) at (7,-1) {$\qquad \qquad \qquad \sstate{\mleq{x}{y}, \mleq{y}{z}, \mleq{z}{y}, \mleq{x}{z}}$};
\node (S2') at (6,-2) {$\sstate{\mleq{x}{y}, \mleq{x}{z}, y\eq z}$};

\draw[forall] (Sc) to node [left] {$anti.$}  (S1);
\draw[forall] (Sc) to node [right] {$\,trans.$}  (S2);
\draw[exists] (S2) to node [right] {$\,anti.$}  (S2');
\draw[exists] (S2') to node [left] {$dupl.$}  (S1);
\end{tikzpicture}
\caption{Some rule-decreasing critical peaks for $\Prog$}
\label{figure:example:leq}
\figrule
\end{figure}

\end{example}

\begin{example}
\newcommand\frk{\textbf{frk}}%
\newcommand\eat{\textbf{eat}}%
\newcommand\thk{\textbf{thk}}%
\renewcommand\frk{\textbf{f}}%
\renewcommand\eat{\textbf{e}}%
\renewcommand\thk{\textbf{t}}%
\label{example:diagrams:philos}%
\renewcommand\Prog{\mathcal P_{\labelExamplePhilos}}%
Consider the program $\Prog$ implementing the dining philosophers
  problem, as given in \cref{example:chr:philos}. 
The confluence of $\Prog$ cannot be inferred by
either local or strong confluence.  On the one hand, $\Prog$ is
obviously non-terminating, and hence prevents
 the application of the local confluence criterion. On the other
hand, $\Prog$ has critical peaks which are not in $(\raps\!\!\cdot\laps)$.
Consider as an example the peak given in \cref{figure:example:philos}.
It is critical for the rule $eating$ with itself, but it is not
joinable by $(\rape\!\! \cdot \lape)$. %
However, the figure shows that it is joinable by \[
\rape_{thk}\cdot \rape_{eat}\cdot \rape_{thk}\cdot \lape_{thk}\cdot
\lape_{eat}\cdot \lape_{thk}\]%
\ie\ the peak is decreasing. In fact, all the critical peaks of $\Prog$
involve only the rule $eat$ and may be closed in a similar manner.
Thus, by assuming that the $eat$ rule is coinductive and
strictly greater than $thk$, we can infer, using
\cref{thm:chr:critical}, that $\Prog$ is confluent.

\begin{figure}
\figrule
\begin{tikzpicture}[xscale=1,yscale=0.95,
forall/.style={->}, 
exists/.style={densely dotted,->}]
\def\xA{3.5} 
\def\yA{-1.5} 
\def\yB{-3} 
\def\yC{-4.5} 
\def\yD{-6} 

\node (Sc) at (0,0) {$\sstate{ \frk{(x)},\frk{(y)}, \frk{(z)},\thk{(x,y,i)},
  \thk{(y,z,j)}}$};

\node (S1) at (-\xA, \yA)
  {$\sstate{\frk{(z)}, \eat{(x,y,i+1)}, \thk{(y,z,j)}}$} ;
\node (S1') at (-\xA, \yB) 
   {$\sstate{ \frk{(x)},\frk{(y)}, \frk{(z)},\thk{(x,y,i+1)},
  \thk{(y,z,j)}}$};
\node (S1'') at (-\xA,\yC) {$\sstate{\frk{(x)}, \thk{(x,y,i+1), \eat{(y,z,j+1)}}}$};

\node (S2) at (\xA,\yA) {$\sstate{\frk{(x)}, \thk{(x,y,i), \eat{(y,z,j+1)}}}$};
\node (S2') at (\xA, \yB) 
   {$\sstate{ \frk{(x)},\frk{(y)}, \frk{(z)},\thk{(x,y,i)},
  \thk{(y,z,j+1)}}$};
\node (S2'') at (\xA, \yC)
  {$\sstate{\frk{(z)}, \eat{(x,y,i+1)}, \thk{(y,z,j+1)}}$} ;

\node (Sj) at (0, \yD)
{$\sstate{ \frk{(x)},\frk{(y)}, \frk{(z)},\thk{(x,y,i+1)},
    \thk{(y,z,j+1)}}$};

\draw[forall] (Sc) to node [above left] {$eat$}  (S1);
\draw[forall] (Sc) to node [above right] {$eat$}  (S2);
\draw[exists] (S1) to node [left] {$thk$} (S1');
\draw[exists] (S1') to node [left] {$eat$} (S1'');
\draw[exists] (S1'') to node [left] {$thk$} (Sj);
\draw[exists] (S2) to node [right] {$thk$} (S2');
\draw[exists] (S2') to node [right] {$eat$} (S2'');
\draw[exists] (S2'') to node [right] {$thk$} (Sj);
\end{tikzpicture}
\caption{A rule-decreasing critical peak of $\Prog$}
\label{figure:example:philos}
\figrule
\end{figure}
\end{example}

\subsection{On program partitioning} 
\label{subsection:diagrams:partition}

The rule-decreasingness criterion is based on the division of the
program into a terminating part and a possibly non-terminating one.
Since a program can be partitioned in multiple ways, it may be the
case that the rule-decreasingness of a program depends on the
splitting used (see \cref{example:chr:counter}).  From a purely
theoretical point of view, this is not a particular drawback, since
the property we aim at proving (i.e. the confluence of program) is
undecidable.
 From
a more pragmatical point of view, it appears that the classic
examples of CHR programs can be proved to be rule-decreasing without any
assumption of termination. In particular, we were unable to find a counterexample of a confluent but non-strongly rule-decreasing program in
\citeANP{Fruehwirth09cambridge}'s book
\citeyear{Fruehwirth09cambridge}.

\begin{example}
  Consider the CHR solver for partial order given in
  \cref{example:leq}. Assuming that any rule is coinductive, $\Prog$ can
  be shown strongly rule-decreasing with respect the order $\succeq$
  satisfying:
\[
\textit{transitivity} \succ  \textit{duplicate}  \succ  \textit{antisymmetry} \succ \textit{reflexivity}
\]
As illustrated by \Cref{figure:example:leq}, critical peaks involving
$\textit{transitivity}$ rules may be closed using only rules that are
strictly smaller. Similarly, one can verify that any critical peak
between a given rule $\alpha$ and a smaller (or equal) one can be
closed using only rules strictly smaller than $\alpha$ (\ie\ all the peaks
are trivialy decreasing).
\end{example}

The choice of a good partition may simplify proofs of
rule-decreasingness: by maximizing the inductive part of a program,
the number of peaks which must be proved decreasing (\ie\ 
the coinductive critical peaks) is reduced.
Indeed, while the joinability of a peak with respect to the inductive
part of program -- which must be terminating -- is a decidable problem and can be efficiently
automatized,\footnote{See the works about CHR local confluence%
~\cite{AFM99constraints,Abdennadher97cp}.} the
rule-decreasingness of a peak with respect to
a possibly non-terminating program is likely to be
undecidable.\footnote{Decreasingness of a peak for a given order seems
  a more difficult problem than joinability without termination
  assumption---which is itself undecidable.} Consequently, a good
partition will limit the use of heuristics or human interactions
necessary to infer a rule-decreasing diagram for each coinductive
critical peak.

Since termination is also an undecidable property, we
cannot expect to fully automatize the search for the optimal partition,
and we must content ourselves with heuristic procedures. 
Despite the fact that the formal development of such procedures is
beyond the scope of
this paper, our practical experience suggests that a trivial partitioning may be
interesting.  This partition consists of  considering as inductive only
those rules that
strictly reduce the number of atoms in a state.
Even if this
choice is not necessarily optimal and may even produce bad
partitions, it does seem to produce relevant partitions for
typical CHR solvers, as illustrated by \cref{example:diagrams:leq}.

\medskip 

We now give two counterexamples. The first shows that
rule-decreasingness can be dependent on particular splittings, while
the second presents a confluent program which is not rule-decreasing.

\begin{example}%
\newcommand\Pp{\textbf{p}}%
\newcommand\Fs{\textup{s}}%
\label{example:chr:counter}%
\xdef\labelExampleCounter{\thetheorem}%
\renewcommand\Prog{\mathcal P_{\labelExampleCounter}}%
Consider the following {CHR} rules:
\[
\textit{duplicate} \at \Pp(x) \backslash \Pp(x) \!\simplif\! \top 
 \qquad
\Fs^- \at \Pp(\Fs(x))\! \simplif \!\Pp(x) 
 \qquad
\Fs^+ \at \Pp(x)\! \simplif\! \Pp(\Fs(x)) 
\] 
We denote by $\Prog^-$ the program built from the $\textit{duplicate}$
and $\Fs^-$ rules, and by $\Prog^+$ the program built from 
the $\textit{duplicate}$ and $\Fs^+$ rules. %

$\Prog^-$ is clearly terminating: the $\textit{duplicate}$ rule strictly
reduces the number of atoms in a state, while $\Fs^-$ leaves the number of
atoms unchanged, but strictly reduces the size of the argument of one
of them.  We can also verify that  $\Prog^-$ has a single critical
peak. \cref{figure:example:counter1} shows the only way this peak may 
be closed. %
Thus, by assuming that all rules are inductive, we can infer that the
program is rule-decreasing. However if $\Fs^-$ is assumed to be
coinductive, we can verify that the sole critical peak of $\Prog^-$ is
decreasing with respect to no admissible order.

\begin{figure}
\figrule
\begin{minipage}[b]{0.49\linewidth}\centering%
\begin{tikzpicture}[xscale=1.5,yscale=1,
forall/.style={->},
exists/.style={densely dotted,->}]

\node (Sc) at (0, 0.4)  {$\sstate{\Pp(\Fs(x)), \Pp(\Fs(x))}$}; 
\node (S1) at (-1, -1)   {$\sstate{\Pp(\Fs(x))}$};
\node (S2) at (1, -1)   {$\sstate{\Pp(\Fs(x)), \Pp(x)}$};
\node (S2') at (1, -2.3) {$\sstate{\Pp(x), \Pp(x)}$};

\node (Sj) at (-1, -2.3)    {$\sstate{\Pp(x)}$};

\draw[forall] (Sc) to node [above left] {$\Pp$} (S1);
\draw[forall] (Sc) to node [above right] {$\Fs^-$}  (S2);
\draw[exists] (S2) to node [right]  {$\Fs^-$}  (S2');
\draw[exists] (S1) to node [left] {$\Fs^-$} (Sj);
\draw[exists] (S2') to node [above] {$\Pp$}  (Sj);
\end{tikzpicture}
\caption{Critical peak of $\Prog^-$}
\label{figure:example:counter1}
\end{minipage}
\begin{minipage}[b]{0.49\linewidth}\centering%
\begin{tikzpicture}[xscale=1.6,yscale=1,
forall/.style={->},
exists/.style={densely dotted,->}]

\node (Sc) at (0, 0.4)  {$\sstate{\Pp(x), \Pp(x)}$}; 
\node (S1) at (-1, -1)   {$\sstate{\Pp(x)}$};
\node (S2) at (1, -1)   {$\sstate{\Pp(\Fs(x)), \Pp(x)}$};
\node (S2') at (1, -2.3) {$\sstate{\Pp(\Fs(x)), \Pp(\Fs(x))}$};

\node (Sj) at (-1, -2.3)    {$\sstate{\Pp(\Fs(x))}$};

\draw[forall] (Sc) to node [above left] {$\Pp$} (S1);
\draw[forall] (Sc) to node [above right] {$\Fs^+$}  (S2);
\draw[exists] (S2) to node [right]  {$\Fs^+$}  (S2');
\draw[exists] (S1) to node [left] {$\Fs^+$} (Sj);
\draw[exists] (S2') to node [above] {$\Pp$}  (Sj);
\end{tikzpicture}
\caption{Critical peak of $\Prog^-$}
\label{figure:example:counter2}
\end{minipage}
\figrule

\end{figure}

As in the case of $\Prog^-$, $\Prog^+$ yields only one critical peak which is
decreasing with respect to no admissible order (see
\cref{figure:example:counter2}). However, this time $\Fs^+$ is not
terminating, and so cannot been assumed inductive. Consequently $\Prog^+$
cannot be inferred to be confluent using \cref{thm:chr:critical}.

\end{example}

\section{Modularity of  {CHR} confluence}
\label{section:modularity}

In this section, we are concerned with proving the confluence of union of
confluent programs in a modular way (in particular of those programs proved
confluent using the rule-decreasing criterion). In practice, we
improve on a result of \linebreak \citeN{Fruehwirth09cambridge} which states that a
terminating union of confluent programs which do not overlap (\ie\
which do not have a critical peak) is confluent. In particular, we
allow some overlapping and we drop the termination hypotheses.

\begin{theorem}[Modularity of confluence]
\label{theorem:modularity}
Let $\Prog$ and $\Q$ be two confluent {CHR} programs.  If any critical
peak between $\Prog$ and $\Q$ is joinable by $\raps[\Q] \cdot \lape[\Prog]$,
then $\Prog\Q$ is confluent.
\end{theorem}

Before formally proving the theorem, it is worth noting that, despite
the fact that modularity of confluence and the rule-decreasing theorem have similar
flavors, both results have different scopes.
Indeed, on the one hand modularity of confluence does not assume anything about the way in which $\Prog$ and $\Q$ are confluent. For instance, if $\Prog$
and $\Q$ are two rule-decreasing programs, \cref{theorem:modularity}
does not require the union of the inductive parts of $\Prog$ and $\Q$ to
be terminating, while \cref{thm:chr:critical} does.
This is important since, 
termination is not a modular property: even if two terminating
programs do not share any user-defined atoms, one cannot be sure that
their union is terminating.  (See %
Section 5.4 of \citeANP{Fruehwirth09cambridge}'s book
\citeyear{Fruehwirth09cambridge} for more details.)
On the other hand, the rule-decreasing criterion allows the critical peaks to
be closed in a more complex way than  \cref{theorem:modularity} permits.

\medskip

The proof of the theorem rests on the following lemma, which states that
under the hypotheses of \cref{theorem:modularity}, $\rap$ ``strongly
commutes'' with $\rap[\Q]$.

\begin{lemma}
\label{lemma:modularity}
  If critical peaks between $\Prog$ and $\Q$ are in $\raps[\Q] \!\cdot
  \lape[\Prog]$, then $(\lap[\Prog]\!\cdot\raps[\Q])\subset
  (\raps[\Q]\!\cdot\lape[\Prog])$.
\end{lemma}

\begin{proof}
  We prove by induction on the length of the derivation $S_c \raps[\Q]
  S'$ that for any peak $S \lap[\Prog] S_c \raps[\Q] S'$, the property
  $S \raps[\Q] \cdot\lape[\Prog] S'$ holds.  The base case $S_c \equiv
  S'$ is immediate. For the inductive case $S \lap[\Prog] S_c \raps[\Q]
  S'' \rap[\Q] S'$, we know by the induction hypothesis that there exists a
  state $R$, such that $S \raps[\Q] R \lape[\Prog] S''$. From here, it
  is sufficient to prove that $R \raps[\Q] \cdot \lape[\Prog] S'$ and to use
  the definition of relation composition in order to conclude.  We assume that $S''
  \rap[\Q] R$, otherwise $R \raps[\Q] \cdot \lape[\Prog] S'$ holds
  trivially.  We distinguish two cases: either the rules involved in
  the local peak $R \lap[\Prog] S'' \rap[\Q] S'$ apply to different parts
  of $S''$, or else their applications overlap.  In the first case, we use
  {CHR} monotonicity to infer $R \rap[\Q] \cdot \lap[\Prog] S''$.  In the
  second case, there must exist a critical peak $R'' \lap[\Prog] \cdot
  \rap[\Q] S'''$, a state $R'$, and a set of variables $\vec x$, such
  that $R'' \oplus_{\vec x} R' \equiv R$, $S''' \oplus_{\vec x}
  R'\equiv S'$. Then by the hypotheses and {CHR} monotonicity, we obtain the results that $R \raps[\Q] \cdot \lape[\Prog] S'$.
\end{proof}

\begin{proof}[Proof of \cref{theorem:modularity}]
  Let $\mathord{\ra_1} = \mathord{\raps[\Prog]}$, $\mathord{\ra_2} =
  \mathord{\raps[\Q]}$. On one hand, by the confluence of $\Prog$ and
  $\Q$, we have $(\mathord{\la_1}\!\cdot\mathord{\ra_1})\subset
  (\mathord{\ra_1}\!\cdot\mathord{\la_1})$ and
  $(\mathord{\la_2}\!\cdot\mathord{\ra_2})\subset
  (\mathord{\ra_2}\!\cdot\mathord{\la_2})$. (Note that
  $\mathord{\ras_1} = \mathord{\ra_1}$ and $\mathord{\ras_2} =
  \mathord{\ra_2}$.)  On the other hand, by combining
  \cref{lemma:modularity} and case {\it (i)} of \cref{lemma:terese},
  we infer $(\mathord{\la_1}\!\cdot\mathord{\ra_2})\subset
  (\mathord{\ra_2}\!\cdot\mathord{\la_1})$. By a trivial application
  of \cref{thm:decreasing}, we find that $\ra_{\{1,2\}}$ is
  confluent. We conclude by noting $\mathord{\rap[\Prog\Q]}
  \subset\mathord{ \ra_{\{1,2\}}} \subset\mathord{\raps[\Prog\Q]}$, and
  apply case {\it(ii)} of \cref{lemma:terese}. (It is worth noting
  that $\mathord{ \ra_{\{1,2\}}}$ equals neither
  $\mathord{\rap[\Prog\Q]}$ nor $\mathord{\raps[\Prog\Q]}$.)
\end{proof}
\section{Conclusion}
\label{section:conclusion}

By employing the decreasing diagrams technique
 in {CHR}, we have
established a new criterion for {CHR} confluence that generalizes
local and strong confluence criteria.  
The crux of this novel criterion rests on the distinction between the
terminating part (the so-called inductive part) and  non-terminating
part (the so-called coinductive part) of a program, together with the
labeling of transitions by rules. %
Importantly, we demonstrate that in the particular case of the
proposed application of the decreasing diagrams, the check
on decreasingness can be restricted to the sole critical pairs, hence
making it possible to automatize the process.
We also improve on a result about the so-called modularity of
confluence, which allows a modular combination of rule-decreasing
programs, without loss of confluence.

It is worth saying that all the diagrammatic proofs sketched in the
paper have been systematically verified by a prototype of a
diagrammatic confluence checker. 
In practice, this checker automatically generates all the critical
pairs of a program provided with an admissible order, then using
user-defined {\em tactics} (finit sets of reductions) tries to
join these while respecting rule-decreasingness.

Current work involves investigating the development of heuristics to
automatically infer rule-decreasingness without human interaction.  We
also plan to develop a new completion procedure based on the criterion
presented here.  %
Because duplicate removal is an important programming idiom
  of CHR, the development of new confluence-proof techniques capable of
  dealing with confluent but non-rule-decreasing programs, like those given in
  \cref{example:chr:counter}, is also worth investigating.

\clearpage 
\bibliographystyle{acmtrans}
\bibliography{biblio}

\end{document}